\documentclass{jgaa-art}

\usepackage{graphicx}
\usepackage{amsmath}
\usepackage{amssymb}
\usepackage{xspace}
\usepackage{subfigure}

\newtheorem{corollary}[theorem]{Corollary}

\usepackage[]{todonotes}
\newcommand{\calA}{{\ensuremath{\cal A}}}
\newcommand{\calB}{{\ensuremath{\cal B}}}
\newcommand{\calC}{{\ensuremath{\cal C}}}

\newif\iffull
\fulltrue

\begin{document}
\title{Drawing outer-1-planar graphs revisited\footnote{Research supported by NSERC.
The author would like to thank Jayson Lynch for pointing out that the lower bound
also holds for IC-planar graphs.}
}
\author[first]{Therese Biedl}{biedl@uwaterloo.ca}

\affiliation[first]{David R.~Cheriton School of Computer
Science, University of Waterloo, Waterloo, Ontario N2L 1A2, Canada. }



\maketitle

\begin{abstract}
In a recent article (Auer et al, Algorithmica 2016) it was claimed
that every outer-1-planar graph has a planar visibility representation
of area $O(n\log n)$.  In this paper, we show that this is
wrong: There are outer-1-planar graphs that require $\Omega(n^2)$
area in any planar drawing.  
Then we give a construction (using crossings, but preserving a
given outer-1-planar embedding) that results in an orthogonal
box-drawing with $O(n\log n)$ area and at most two bends
per edge.
\end{abstract}

\section{Introduction}

A {\em 1-planar} graph is a graph that can be drawn in the plane
such that every edge has at most one crossing.  
Many graph-theoretic and graph-drawing results are 
known for 1-planar graphs, see for example \cite{KLM17}.  One
subclass 
\iffull 
of 1-planar graphs 
\fi
is the class of {\em outer-1-planar (o1p) graphs},
which have a 1-planar drawing such that
additionally every vertex is on the outer-face (the unbounded
region of the drawing).

Outer-1-planar graphs were introduced by Eggleton \cite{Eggleton}
and studied by many other researchers \cite{Argyriou, Auer16, DE12,HEL+12}.
Of particular interest to us is a paper by 
Auer, Bachmeier, Brandenburg, Glei{\ss}ner, Hanauer, Neuwirth and
Reislhuber \cite{Auer16}.  Among others, they characterize the
forbidden minors of outer-1-planar graphs, give a recognition
algorithm, and give bounds on various
graph parameters such as number of edges, treewidth, stack number and queue number.
Finally they turn to drawing algorithms for outer-1-planar graphs,
and here claim the following result:  ``Every o1p graph has a planar 
visibility representation in $O(n \log n)$ area.''  (Theorem 8).

In this paper, we show that this result is incorrect.  Specifically,
we construct an $n$-vertex outer-1-planar graph such that
in {\em any} planar embedding there are $\Omega(n)$ nested
triangles (we give detailed definitions below).
It is known \cite{FPP88} that any planar graph drawing with 
$k$ nested cycles requires width and height at least
$2k$ in any planar poly-line drawing.  Since any planar visibility
representation can be converted into a poly-line drawing of
asymptotically the same width and height \cite{Bie-GD14}, 
any planar visibility representation of our graph uses $\Omega(n^2)$ 
area and the
claim by Auer et al.~is incorrect.

Then we give drawing algorithms for outer-1-planar graphs that do
achieve area $o(n^2)$.  These drawings have crossings, but we can
reflect exactly the given outer-1-planar embedding.  Our construction gives
orthogonal box-drawings with area $O(n\log n)$ and at most two bends per edge; they can
be converted to poly-line drawings of the same area.

To our knowledge, the only prior result on orthogonal drawings of
outer-1-planar drawings (other than the one by Auer et al.~that
we disprove) is by Argyriou et al.~\cite{Argyriou}; they showed that 
every outer-1-planar graph with maximum degree 4 has an outer-1-plane
point-orthogonal drawing with $O(n^2)$
area and at most 2 bends per edge.   

\section{Definitions}

We assume familiarity with graphs, see e.g.~\cite{Die12}.
A {\em planar graph} is a graph that can be drawn in the plane
without any crossing.    Such a drawing $\Gamma$ defines the
{\em regions}, which are the connected parts of $\mathbb{R}^2
\setminus \Gamma$.  The infinite region is called the {\em outer-face}.
A planar drawing defines the {\em planar embedding}
consisting of the {\em rotation scheme} (the clockwise order of edges 
at each vertex) and the {\em outer-face} (a lists of
vertices and edges on the outer-face).  A graph is called
{\em outer-planar} if it has a planar embedding where all vertices
are on the outer-face.

A {\em 1-planar graph} is a graph that can be drawn in the plane
such that every edge has at most one crossing.  
As above one defines {\em regions} and {\em outer-face} of such a drawing. 
An {\em outer-1-planar} graph is a graph with a 1-planar drawing
where additionally all vertices are on the outer-face.  Any
such drawing is described via an {\em outer-1-planar embedding},
consisting of the rotation scheme,
the outer-face, and information
as to which pair of edges cross.  

In this paper we almost only consider {\em maximal outer-planar}
and {\em maximal outer-1-planar} graphs, which are those graphs
where as many edges as possible have been added while staying in
the same graph class and having no duplicate edges or loops.

A {\em poly-line drawing} of a graph is a drawing where vertices 
are points and edges are polygonal curves; a {\em bend}
is the transition-point between segments of the polygonal curve.
We also consider {\em orthogonal
box-drawings}, where vertices are represented by axis-aligned
boxes and edges are polygonal curves with
only horizontal and vertical segments.  A special kind of orthogonal
box-drawing is a {\em visibility representation} where edges have 
no bends.

The orthogonal box-drawings created in this paper are somewhat specialized 
in that vertices are {\em flat}: All vertex-boxes are actually
horizontal line segments (in the figures, we show them thickened into
a thin rectangle).    We call such a vertex-box a {\em bar}
and such an orthogonal box-drawing an {\em orthogonal bar-drawing}.

We assume (without further mentioning) that all our drawings are
{\em grid-drawings}, i.e., all defining features (vertex-points,
endpoints of vertex-bars, bends) are placed at points with integer
coordinates.  We measure the {\em width} and {\em height}
of a grid-drawing as the number of vertical/horizontal grid-lines that intersect the 
smallest enclosing bounding box of the drawing.  We call a drawing
{\em order-preserving} if it exactly reflects a given (planar
or 1-planar) embedding of the graph.

\section{Lower bound}

In this section, we construct
an outer-1-planar graph that requires $\Omega(n^2)$ area in any
planar poly-line drawing. 
Our graph $G_L$ (for $L\geq 2$ even)
consists of a $2\times L$-grid with every second region filled
with a crossing.  Clearly this is an outer-1-planar graph, see 
Figure~\ref{fig:lower}.  
Enumerate the vertices of $G_L$ as in the figure.

\begin{figure}[ht]
\hspace*{\fill}
\includegraphics[width=0.5\linewidth,page=3]{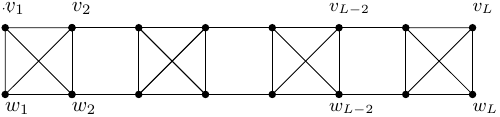}
\hspace*{\fill}
\includegraphics[width=0.35\linewidth,page=4]{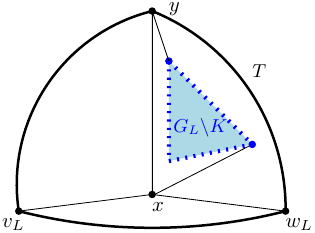}
\hspace*{\fill}
\caption{The outer-1-planar graph $G_8$, and how to find nested triangles.}
\label{fig:lower}
\end{figure}

It is known that all outer-1-planar graphs have a planar
drawing \cite{Auer16}, but
they can have many different planar embeddings. 
However, we can show that for our graph $G_L$, all
planar embeddings are bad in some sense.  

Call a set of disjoint triangles $T_1,\dots,T_\ell$ 
{\em nested} (in a fixed planar embedding) 
if for $i=2,\dots,\ell$ the region bounded by $T_i$ includes 
all vertices of $T_1\cup \dots \cup T_{i-1}$.

\begin{lemma}
\label{lem:lower}
Fix $L\geq 2$ even.  Any planar embedding $\Gamma$ of $G_L$
with $(v_L,w_L)$ on the outer-face
contains $L/2$ nested triangles.
\end{lemma}
\begin{proof}
Set $K:=\{v_L,w_L,v_{L-1},w_{L-1}\}$.   
These four vertices form a $K_4$; its induced embedding $\Gamma_K$ 
is hence unique up to renaming.  By assumption the outer-face $T$ of $\Gamma_K$ 
contains $v_L,w_L$  and one vertex $y\in \{v_{L-1},w_{L-1}\}$; 
set $x=\{v_{L-1},w_{L-1}\}\setminus y$.

If $L=2$, then we are done (use triangle $T$).
If $L>2$, then graph
$G':=G_L\setminus K$ is connected, so must reside entirely within 
one face $f$ of $\Gamma_K$.
Graph $G'$ contains neighbours of $x$ and $y$, so face $f$
must contain both $x$ and $y$.  Since $x$
is not on the outer-face of $\Gamma_K$, face $f$ is not the
outer-face of $\Gamma_K$.  So no  vertex of $G_L\setminus K$ is
in the outer-face of $\Gamma_K$, making $T$ the
outer-face of the entire drawing $\Gamma$. 

Observe that $G'=G_L\setminus K$ is a copy of $G_{L-2}$.
Since both $v_{L-2}$ and $w_{L-2}$ have neighbours in $\{x,y\}$,
edge $(v_{L-2},w_{L-2})$ is on the outer-face of the induced drawing $\Gamma'$
of $G'$.  By induction, $\Gamma'$ contains $L/2-1$
nested triangles $T_1,\dots,T_{L/2-1}$.
Adding the outer-face $T$ to this
gives the desired set of nested triangles for $G$
since $G_{L-2}$ resides inside $T$ and is vertex-disjoint from it.
\end{proof}

\begin{theorem}
\label{thm:lower}
There exists an $n$-vertex outer-1-planar graph that requires width and
height at least $n/4$ in any planar poly-line grid-drawing.
\end{theorem}
\iffull
\begin{proof}
Fix an arbitrary integer $N$, and consider
graph $G_{4N}$ which has $n=8N$ vertices.  Observe that $G_{4N}$
contains two disjoint copies of $G_{2N}$, obtained by removing
the edges $(v_{2N},v_{2N+1})$ and $(w_{2N},w_{2N+1})$.
In any planar embedding of $G_{4N}$,
at least one of these two copies of $G_{2N}$ must have its rightmost/leftmost grid-edge on the outer-face of its induced planar embedding.   
In this copy, we therefore have $N$
nested triangles by Lemma~\ref{lem:lower}.  It is known \cite{FPP88} that 
$\ell$ nested triangles require width and height $2\ell$
in any planar poly-line drawing, which implies the result by
$2N=n/4$.
\end{proof}

If we use so-called {\em 1-fused stacked
triangles} \cite{Bie-DCG11}, then the lower bound can be
improved ever-so-slightly to $(n+2)/4$ after inserting a crossing
into {\em all} inner regions of the $2\times L$-grid; see
a preliminary version of this paper \cite{Bie-GD20} for details.  
We gave the weaker bound here because $G_L$ has two other
advantages:  it is {\em IC-planar} (no two crossings have a common
vertex) and it has maximum degree 4 (so the lower bound even
holds for orthogonal point-drawings).
\fi

\subsection{Drawing outer-planar graphs, and the approach of \cite{Auer16}}
\label{sec:BiedlReview}

We now briefly review the algorithm by Auer et al.~\cite{Auer16} to
explain where the error lies.   Their algorithm is based on
a prior algorithm (we call it here {\sc MaxOutpl}) by the author
that creates an order-preserving orthogonal bar-drawing of any maximal
outer-planar graph \cite{Bie-GD02,Bie-DCG11}.  The idea {\sc MaxOutpl}
is to fix one {\em reference-edge} $(s,t)$ with {\em poles} $s,t$ on the outer-face 
(with $s$ before $t$
in clockwise order). Then draw graph $G$ such that the bars of $s$
and $t$ occupy the top right and bottom right corner respectively.

To do so, split the graph and recurse, see also Figure~\ref{fig:BiedlReview}(a).
Specifically, consider the interior face incident to $(s,t)$ (say its
third vertex is $x$).  Of the two subgraphs ``hanging'' at the
edges $(s,x)$ and $(x,t)$, pick the smaller one.
(Formally, for any edge $(u,v)\neq (s,t)$, the
{\em hanging subgraph} $H_{uv}$ is the graph induced by all outer-face
vertices between $v$ and $u$, using the path from $v$ to $u$ that does {\em 
not} include edge $(s,t)$.)
Assume that $H_{x,t}$ is not bigger than $H_{s,x}$, but has at least three
vertices (all other cases are handled symmetrically or with another simpler
construction).  Let $\{x,y_1,t\}$
be the other interior face at $(x,t)$.  
Recursively draw the
three subgraphs $H_{s,x}$, $H_{x,y_1}$ and $H_{y_1,t}$ with respect
to reference-edges $(s,x),(x,y_1)$ and $(y_1,t)$.  After a minor modification 
of the drawings (``releasing'' one pole, defined
below) and rotating the drawing of $H_{x,y_1}$, 
these drawings can be merged as shown in Figure~\ref{fig:BiedlReview}(b).

\begin{figure}[ht]
\hspace*{\fill}
\subfigure[~]{\includegraphics[page=4,width=0.38\linewidth]{biedl_construction.pdf}}
\hspace*{\fill}
\subfigure[~]{\includegraphics[page=1,width=0.48\linewidth]{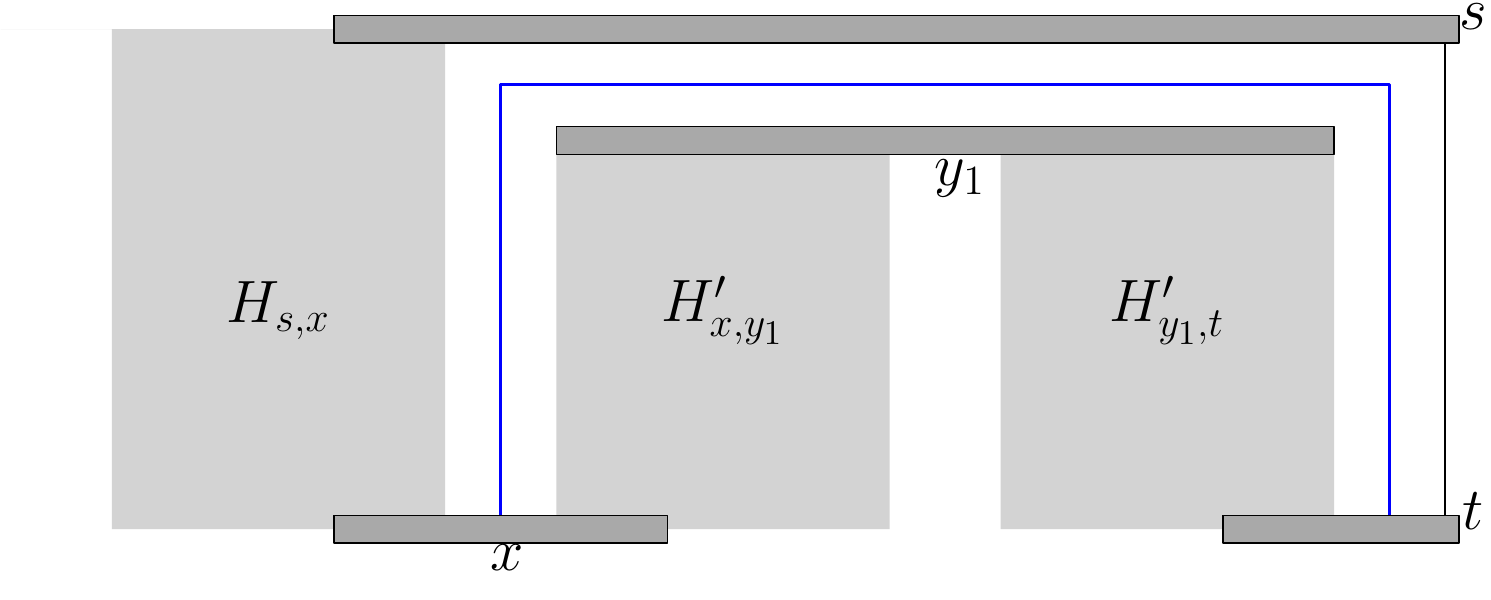}}
\hspace*{\fill} \newline
\hspace*{\fill}
\subfigure[~]{\includegraphics[page=3,width=0.48\linewidth]{biedl_construction.pdf}}
\hspace*{\fill}
\subfigure[~]{\includegraphics[page=2,width=0.48\linewidth]{biedl_construction.pdf}}
\hspace*{\fill}
\caption{The construction by the author \cite{Bie-DCG11} and by
Auer et al.~\cite{Auer16}.  (a) Splitting the graph (the red dashed edge does
not exist for outer-planar graphs).  (b) Putting drawings together in \cite{Bie-DCG11}.
(c) Putting drawings together in \cite{Auer16}.  (d) A variation of \cite{Auer16}.
}
\label{fig:BiedlReview}
\label{fig:AuerReview}
\end{figure}

Auer et al.~\cite{Auer16} used the same idea, but release other poles,
mirror drawings rather than rotate them, 
and route edge $(x,t)$ differently.
This leaves space free to also route edge $(s,y_1)$ and
removes all bends, hence giving a visibility representation.
See Figure~\ref{fig:AuerReview}(c).
However, there are a few issues with this construction:
\begin{itemize}
\item First, the logarithmic height-bound for {\sc MaxOutpl}
	crucially requires that 
	the constructed drawing is no bigger than the
	drawing of the bigger subgraph $H_{s,x}$.  This is violated
	in the construction from Figure~\ref{fig:AuerReview}(c), though
	the issue can easily be fixed by drawing one edge horizontally instead,
	see Figure~\ref{fig:AuerReview}(d).
\item Second, Auer et al.~silently assume that the region incident to $(s,t)$
	has a crossing.  If it does not, but if the
	other region incident to $(x,t)$ has a crossing,
	then it is not even clear how $y_1$ should be picked,
	and the crossing edges are not both drawn.    
\item Finally, even if the region at $(s,t)$ has
	a crossing, the crossing edge need not be $(s,y_1)$.  (Recall
	that in {\sc MaxOutpl} vertex
	$y_1$ is determined by the size of the subgraphs and cannot
	be picked arbitrarily.)  Instead, the crossing edge could connect
	$t$ to a vertex in $H_{s,x}$, and we cannot add such an edge to
	the drawing without adding bends or going through other bars.
\end{itemize}

The third issue is the one that led to our counter-example, constructed
such that if we pick $\{s,t,x,y_1\}$ to be the endpoints of a crossing,
then graph $H_{s,x}$ is {\em not} the biggest of the subgraph (and
neither is $H_{y_1,t}$), and so the logarithmic height-bound fails to hold.

\iffull
\section{Constructions}

It should be quite obvious that if we allow crossings and some bends, 
then we {\em can}
create drawings of area $O(n\log n)$ for any outer-1-planar graph $G$.
Specifically, pick an arbitrary maximal outer-planar subgraph $G^-$ of
$G$, and let $\Gamma^-$ be an orthogonal bar-drawing of $G^-$ obtained with 
{\sc MaxOutpl}.  Since every edge
has at most two bends, every region of $\Gamma^-$ has $O(1)$ bends.
As such, any edge of $G\setminus G^-$
(which needs to be drawn through two adjacent regions) can be inserted
with $O(1)$ bends.    We now work on reducing this bound on the number
of bends and show:

\begin{theorem}
\label{thm:main}
Any outer-1-planar graph has an order-preserving orthogonal
box-drawing with at most two bends per edge and $O(n\log n)$
area.
\end{theorem}

It is straight-forward to convert any planar orthogonal box-drawing
into a poly-line drawing while keeping the area asymptotically
the same and adding at most two bends per edge.  See \cite{Bie-GD14}
for details, and note that the same technique works whether the
drawing is planar or not.  Therefore our result implies:

\begin{corollary}
Any outer-1-planar graph has an order-preserving poly-line
drawing with at most four bends per edge and $O(n\log n)$ area.
\end{corollary}

Since we have a constant number of bends per edge, and any
outer-1-planar graph has $O(n)$ edges \cite{Auer16}, we have $O(n)$ vertical
segments in the orthogonal box-drawing.  As such, after deleting
empty columns if needed, the width is automatically $O(n)$.
Therefore all our analysis is focused on the height of the
drawing, which we prove to be in $O(\log n)$.

\subsection{Drawing types} 

Now we prove Theorem~\ref{thm:lower} with a recursive
drawing algorithm.  We roughly follow the idea of {\sc MaxOutpl},
but explicitly distinguish whether the region incident to
$(s,t)$ is crossed or not.
Crucially, we allow more types of drawings for the subgraphs
to achieve fewer bends overall.

We only draw maximal outer-1-planar graphs; one can always
make an outer-1-planar graph maximal by adding edges, and
delete those edges from the obtained drawing later.
It is known that for a maximal outer-planar graph the edges
on the outer-face have no crossing \cite{Eggleton}.

So fix a maximal outer-1-planar graph $G$ with a fixed
outer-1-planar embedding and with reference-edge $(s,t)$.
An orthogonal bar-drawing $\Gamma$ of $G$ is called
\begin{itemize}
\item a drawing of {\em type A} if the bars of $s$ and $t$ occupy
	the top right and bottom right corner of $\Gamma$, respectively
	(this is the same as for \cite{Bie-DCG11});
\item a drawing of {\em type $\overline{B}$} if the bar of $s$ occupies
	the top right corner of $\Gamma$, and the bar of $t$
	occupies the point one row below this corner;
\item a drawing of {\em type $\underline{B}$} if the bar of $t$ occupies
	the bottom right corner of $\Gamma$, and the bar of $s$
	occupies the point one row above this corner;
\item a drawing of {\em type C} if the bars of $s$ and $t$
	occupy the bottom left and bottom right corner of $\Gamma$,
	respectively. 
\end{itemize} 
All drawings that we create are order-preserving.
In particular edge $(s,t)$ must be
drawn clockwise along the boundary of the drawing; 
Figure~\ref{fig:DrawingTypes} shows how it will
be drawn.

\begin{figure}[ht]
\includegraphics[width=0.19\linewidth,page=24]{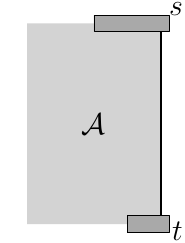}
\hspace*{\fill}
\includegraphics[width=0.19\linewidth,page=25]{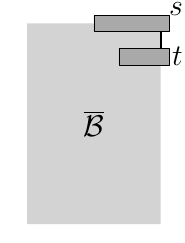}
\hspace*{\fill}
\includegraphics[width=0.19\linewidth,page=26]{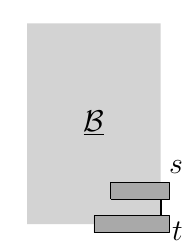}
\hspace*{\fill}
\includegraphics[width=0.19\linewidth,page=27]{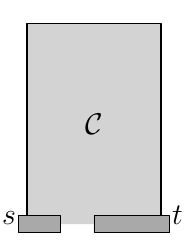}
\hspace*{\fill}
\includegraphics[width=0.15\linewidth,page=29]{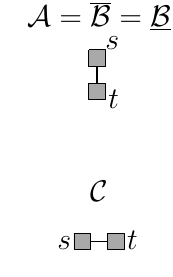}
\caption{The drawing-types, and the base cases.}
\label{fig:DrawingTypes}
\label{fig:base}
\end{figure}

Let $\alpha\approx 0.59$ be such that $\alpha^5=(1-\alpha)^3$.
Let $\phi:=\frac{\sqrt{5}-1}{2}\approx 0.618$ be such that $\phi^2=1-\phi$.
Define 
$\gamma:= \max\{-\frac{2}{\log \phi},- \frac{3}{\log \alpha}\} 
\approx \max\{ 2.88, 3.94 \} = 3.94,$
we hence know
$$\begin{array}{ccc}
\gamma \log \alpha \leq -3, &&
\gamma \log (1{-}\alpha) 
	= \gamma \log (\alpha^{5/3}) 
	= \frac{5}{3} \gamma \log \alpha \leq -5  \\
\gamma \log \phi \leq -2, & \quad &
\gamma \log (1{-}\phi) 
	= \gamma \log(\phi^2) 
	= 2\gamma \log \phi \leq -4, 
\end{array}$$

Also set $\delta=2$.
We measure the {\em size} $|G|$ of an $n$-vertex maximal outer-1-planar graph $G$ 
as $n-1$; this may be rather unusual but will help keep the equations simpler.
Define $h(G):=\gamma \log |G|+\delta \approx 3.94\log(n-1)+2$; this is the height
that we want to achieve in our drawings.  Theorem~\ref{thm:main}
now holds if we show the following result.

\begin{lemma}
\label{lem:main}
Let $G$ be a maximal outer-1-planar graph with reference-edge
$(s,t)$.  Then $G$ has
order-preserving orthogonal bar-drawings with at most two bends
per edge and of the following kind:
\begin{itemize}
\item A type-A drawing $\calA$ of height at most $h(G)$,
\item a type-$\overline{B}$ drawing $\overline{\calB}$ of height at most $h(G)+2$,
\item a type-$\underline{B}$ drawing $\underline{\calB}$ of height at most $h(G)+2$, and
\item a type-C drawing $\calC$ of height at most $h(G)+3$.
\end{itemize}
Furthermore, at least one of $\overline{\calB}$ and $\underline{\calB}$ has height at most $h(G)$.
\end{lemma}

We prove Lemma~\ref{lem:main} by induction on $|G|$.
In the base case, $G$ consists of only edge $(s,t)$, and one easily
constructs suitable drawings, even without bends.  See Figure~\ref{fig:base}.
The height is at most 2 in all cases.  Since $|G|=1$, we have
$\log|G|=0$ and the bound holds by $\delta=2$.

\subsection{Subgraphs and tools}

Now assume that $n\geq 3$, so $G$ has at least one inner region.
The idea is to split $G$ into subgraphs, recursively obtain their
drawings, and put them together suitably.  We have two cases
(see also Figure~\ref{fig:subgraphs}).
In Case $\Delta$, the inner region at $(s,t)$ has no crossing;
by maximality it is hence a triangle, say $\{s,t,x\}$.
We will recurse on the two hanging subgraphs $H_{s,x}$ and $H_{x,t}$, and
use $H_L:=H_{s,x}$ 
and $H_R:=H_{x,t}$ as convenient shortcuts. Observe that $|H_L|+|H_R|=|G|$
since we define the size to be one less than the number of vertices.
In Case $\times$ the inner region at $(s,t)$ is incident to a crossing, say 
edge $(s,y)$ crosses edge $(t,x)$.  By maximality the edges
$(s,x)$, $(x,y)$ and $(y,t)$ exist and have no crossing. We will recurse on 
the three hanging subgraphs 
$H_L:=H_{s,x}$, $H_M:=H_{x,y}$ and $H_R:=H_{y,t}$.
Observe that $|H_L|+|H_M|+|H_R|=|G|$.
These (two or three) subgraphs
are smaller, and we assume that they have been drawn inductively, giving
drawings $\calA_L,\overline{\calB}_L,\underline{\calB}_L,\calC_L$ for 
subgraph $H_L$, and similarly for the other subgraphs.  In the pictures,
we use \raisebox{5pt}{\rotatebox{180}{$\calA_L$}}
for drawing $\calA_L$ rotated by 180 degrees, and similarly for other drawing-types and subgraphs.

\begin{figure}[ht]
\hspace*{\fill}
\includegraphics[width=0.48\linewidth,page=1]{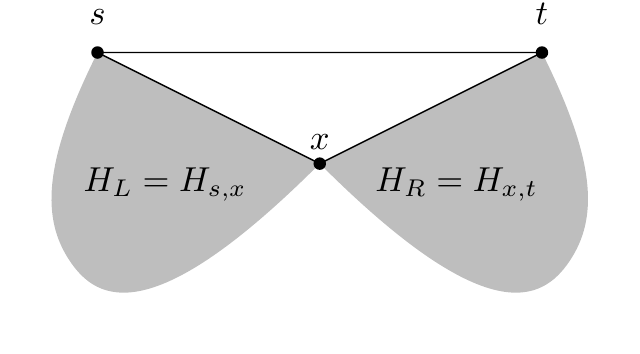}
\hspace*{\fill}
\includegraphics[width=0.48\linewidth,page=2]{subgraph.pdf}
\hspace*{\fill}
\caption{Splitting a subgraph.}
\label{fig:split}
\label{fig:subgraphs}
\end{figure} 

To put drawings together, we frequently use two well-known tools \cite{Bie-GD02,Bie-DCG11}:
\begin{itemize}
\item If we have a drawing $\Gamma$ of some subgraph, then we can insert empty rows to
	increase its height since the drawing is orthogonal.  If we choose the place to
	add empty rows suitably, then this does not change the type of the drawing.
\item If we have a drawing $\Gamma$ of some subgraph, with vertex $s$ in 
	the top row, then we can {\em release} $s$:  add a new row above 
	$\Gamma$, let $s$ occupy all of this row, and re-route edges to 
	neighbours of $s$.   (If $s$ has a horizontal neighbour $z$, then
	the edge $(s,z)$ now becomes vertical.)
	See Figure~\ref{fig:release} and \cite{Bie-DCG11} for details.
	This increases the height by 1, and achieves that
	$s$ now occupies both the top left and top right corner in the resulting
	drawing $\Gamma'$.  

	Similarly we can release vertex $t$ to occupy the bottom-left corner, presuming 
	it is drawn in the bottom row.
	In the pictures, we use a ``prime'' (e.g.~$\calA_L'$ as opposed to $\calA_L$)
	to indicate that one pole has been released; it will be clear
	from the picture which one.    
\end{itemize}

\begin{figure}[ht]
\hspace*{\fill}
{\includegraphics[width=0.69\linewidth,trim=0 10 0 10,clip]{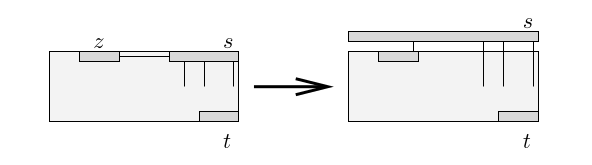}}
\hspace*{\fill}
\caption{Releasing vertex $s$.  Picture taken from \cite{Bie-DCG11}.}
\label{fig:release}
\end{figure}

\subsection{Induction step---Case $\times$}

We start with Case $\times$ where the region incident to $(s,t)$ has 
a crossing, and study the three different types of drawings that we want
to achieve.
We occasionally use $h:=h(G)$ as convenient shortcut.

\medskip\noindent{\bf Case $\times$.A: } 
We want a type-A drawing of height $h$.
We distinguish sub-cases by the size of $H_M$.

\medskip\noindent{\bf Sub-case $\times$.A.1: } $|H_{M}|  \leq \alpha |G|$
(recall that $\alpha\approx 0.59$).
We know that $|H_L|+|H_R|\leq |G|$, hence we may assume $|H_R|\leq |G|/2$ 
and use construction $\times.A.(a)$ from Figure~\ref{fig:timesA}.
(The case $|H_L|\leq |G|/2$ is symmetric and uses construction $\times.A.(b)$.)

\begin{figure}[ht]
\subfigure[][]{\includegraphics[scale=0.38,page=1]{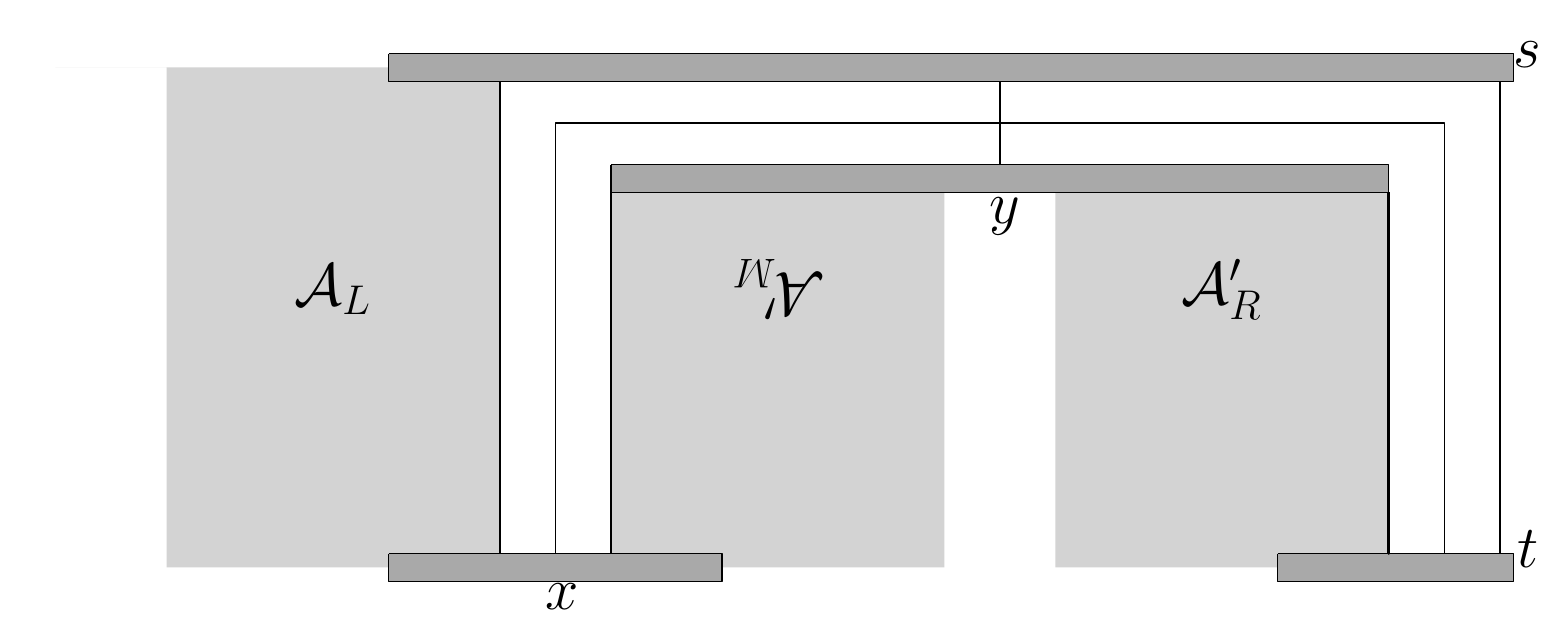}}
\hspace*{\fill}
\subfigure[][]{\includegraphics[scale=0.38,page=2]{construction.pdf}}
\newline
\newline
\subfigure[~]{\includegraphics[scale=0.38,page=3]{construction.pdf}}
\hspace*{\fill}
\subfigure[~]{\includegraphics[scale=0.38,page=4]{construction.pdf}}
\caption{Constructions for $\times.A$.}
\label{fig:timesA}
\end{figure}

We will (for this case only) explain in detail how this figure is to
be interpreted; for later cases we hope that the figures alone suffice.  
We use drawings $\calA_{L}, \calA_{M}$ and $\calA_{R}$ of the 
subgraphs $H_L,H_M,H_R$.
The primes in the
figure indicate that we should release release $y$ in both $\calA_{M}$
and $\calA_{R}$ to get $\calA_M'$ and $\calA_R'$.    
Rotate $\calA_M'$ by $180^\circ$ to get $\raisebox{5pt}{\rotatebox{180}{$\calA_M'$}}$.
Increase the height of drawings, if needed, such that $\calA_R'$ and
$\raisebox{5pt}{\rotatebox{180}{$\calA_M'$}}$ have the same height; then
combine the two bars of $y$ into one.  Increase the height of $\calA_M$,
if needed,
so that it is at least two rows taller than the other two drawings.
Then we combine these drawings and route the edges $(s,y), (x,t)$ and
$(s,t)$ as shown in Figure~\ref{fig:timesA}(a).  

One can easily verify that the result is an order-preserving drawing.
To argue that it has height at most $h$, 
the general procedure is as follows.  First study the height of all three drawings of subgraphs,
which must be at most $h$.  Furthermore, at some of these subgraph-drawings 
more rows are needed, for releasing vertices
and/or routing edges and/or other bars.    
If this is the case,
then we must argue that the subgraph-drawing is sufficiently
much smaller ($h-3$ in case $\times.A.1$).  With this the
total height-requirement at most $h$ at this subgraph-drawing, and 
so other parts of the drawing are not forced to increase beyond height $h$.
After arguing this for
all three subgraphs, we therefore know that the height of the 
constructed drawing is at most $h$.

In the specific case here, the height-analysis is done as follows.
Since $|H_L|\leq |G|$, drawing $\calA_L$ has height at most
$h(H_L)\leq h(G)=h$.  
We have $|H_M|\leq \alpha|G|$, so drawing $\calA_M$ has height
at most
\begin{eqnarray*}
 h(H_M) & = & \log |H_M| + \delta \leq \gamma \log(\alpha \cdot |G|)+ \delta 
= \gamma \log|G| + \delta + \gamma \log \alpha \\
& = & h(G) + \gamma \log \alpha \leq h-3
\end{eqnarray*}
by $\gamma \log \alpha \leq -3$.
Since $|H_R|\leq \tfrac{1}{2}|G| < \alpha |G|$, likewise $\calA_R$ has height at most $h-3$.  We need three more rows above $\calA_M$ and $\calA_R$: one to release $y$, one for edge $(x,t)$ and one for the bar of $s$.  So the height requirement is at most $h$ everywhere as desired.

\medskip\noindent{\bf Sub-case $\times$.A.2: } $|H_{M}|  > \alpha |G|$.
We know that one of $\overline{\calB}_M$ or $\underline{\calB}_M$ has height at most $h(G_M)$.
Let us assume that this is $\underline{\calB}_M$, and we then use construction
$\times.A.(c)$ from Figure~\ref{fig:timesA}
(the other case uses construction $\times.A.(d)$ and is similarly analyzed).

Drawing $\underline{\calB}_M$ has height at most $h(G_M)\leq h$ by assumption.
Since $|H_R|\leq |G|-|H_M|\leq (1-\alpha)|G|$, 
drawing $\calC_R$ has height at most 
$$h(H_R)+3 
\leq \gamma \log ((1-\alpha)|G|) + \delta + 3
= h(G) + \gamma \log(1-\alpha) + 3 
\leq h-2$$
by $\gamma\log(1-\alpha) \leq -5$.
We need two more rows above $\calC_R$ (for $(x,t)$ and the bar of $s$) and the
height requirement is at most $h$ here.
Drawing $\calA_L$ has height at most $h(H_L)$, which by $|H_L|\leq (1-\alpha)|G|$
is similarly shown to be at most $h-5$.  We need two rows below $\calA_L$ (for releasing $x$ and the bar of $y$).
Therefore the height requirement is at most $h$ everywhere.

\medskip\noindent{\bf Case $\times$.B: } 
We want two drawings, of type
$\overline{B}$, $\underline{B}$.  Both have height at most $h+2$,
and one has height at most $h$.

Consider first constructions $\times.B.(a)$ for a type-$\overline{B}$ drawing,  
and $\times.B.(b)$ for a type-$\underline{B}$ drawing, see Figure~\ref{fig:timesB}.   
In both, the drawing of
$H_M$ has height at most $h(H_M)+2\leq h+2$.  Drawings $\calA_L$ and $\calA_R$ have height
at most $h$, and we need two more rows at them (one to release a pole, one
for a bar of a vertex not in the subgraph).  So either drawing
has height at most $h+2$ as desired.  

But we must distinguish
cases (and perhaps use a different construction) to achieve
that one of the drawings has height at most $h$.

\begin{figure}[ht]
\subfigure[~]{\includegraphics[scale=0.38,page=8]{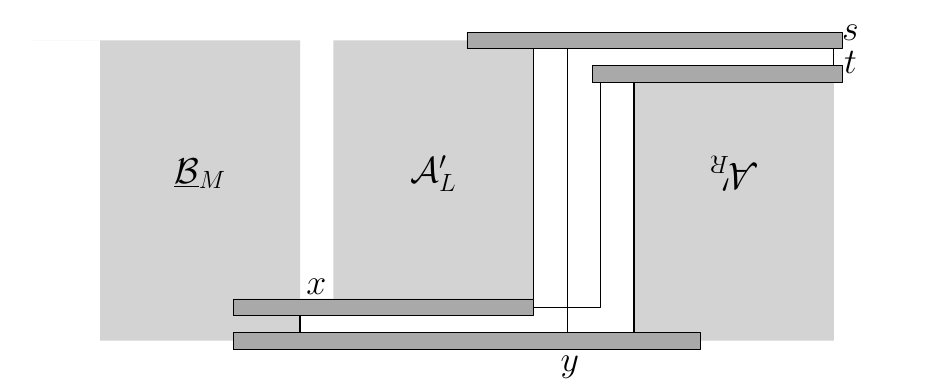}}
\hspace*{\fill}
\subfigure[~]{\includegraphics[scale=0.38,page=14]{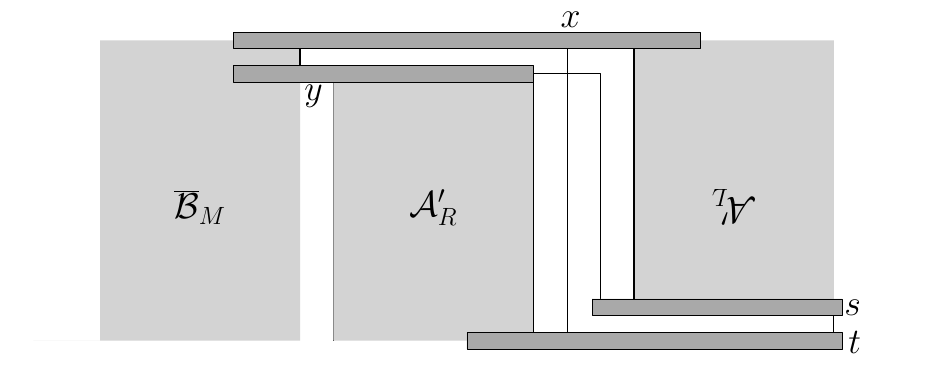}}
\newline
\newline
\subfigure[~]{\includegraphics[scale=0.38,page=15]{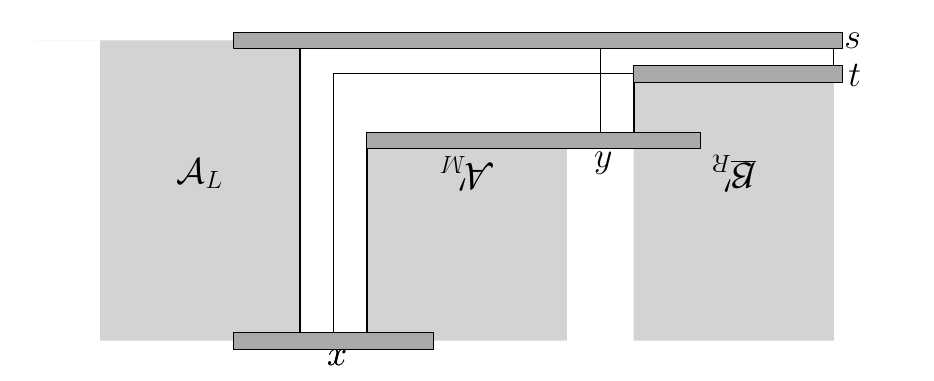}}
\hspace*{\fill}
\subfigure[~]{\includegraphics[scale=0.38,page=16]{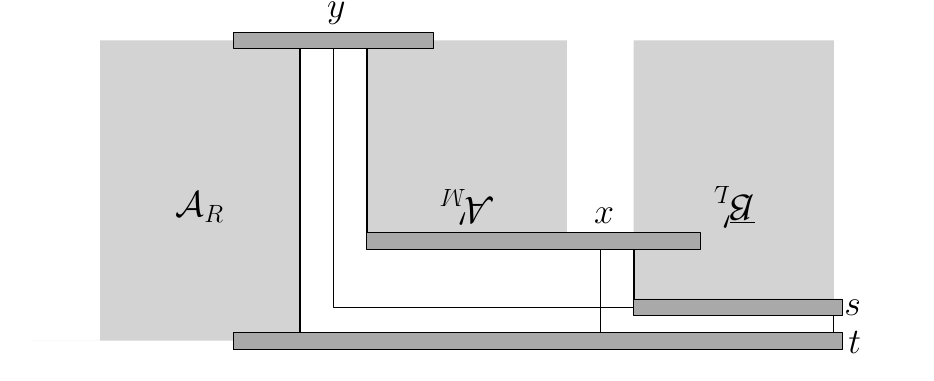}}
\caption{Constructions for $\times.\overline{B}$ and $\times.\underline{B}$.}
\label{fig:timesB}
\end{figure}

\medskip\noindent{\bf Sub-case $\times$.{B}.1: } $|H_L|,|H_R| \leq \phi|G|$
(recall that $\phi=(\sqrt{5}-1)/2\approx 0.618$).
We know that one of $\overline{\calB}_M$ or $\underline{\calB}_M$ has height at most $h(G_M)$.
Let us assume that this is $\underline{\calB}_M$, and consider again construction
$\times.B.(a)$ (in the other case one similarly analyzes construction $\times.B.(b)$).

Drawing $\underline{\calB}_M$ by assumption has height at most $h$.  
Also for $i\in \{L,R\}$ we have $|H_i|\leq \phi|G|$ and drawing $\calA_i$ has height at most
$$ h(H_i)= \gamma \log |H_i| + \delta 
\leq \gamma \log|G| + \delta + \gamma \log \phi
\leq h(G) + \gamma \log \phi \leq h-2$$
since $\gamma\log \phi \leq -2$.
We need two further rows at each of $\calA_L$ and $\calA_R$, so the height 
requirement is at most $h$ everywhere. 

\medskip\noindent{\bf Sub-case $\times$.{B}.2: } $|H_L|>\phi|G|$.
Use construction $\times.B.(c)$ to obtain a type-$\overline{B}$ drawing,
see Figure~\ref{fig:timesB}.
We know that $|H_R|\leq (1-\phi)|G|$ and hence $\underline{\calB}_R$ has
height at most
$$ h(H_R)+2 
\leq h(G) + 2 + \gamma \log(1-\phi) \leq h-2$$
since $\gamma\log (1-\phi) \leq -4$.
We need two more rows above it (one to release $t$ and one for the bar
of $s$), so the height requirement at $\calB_R$ is at most $h$.  Similarly the height
of $\calA_M$ is at most $h(H_M) \leq h-4$.  We need
four more rows above it: one row for releasing $y$, one row because
$\calB_R'$ has a row between $y$ and the (released) $t$, one row
for $(x,t)$ and one row for the bar of $s$.
So the height requirement is at most $h$ everywhere.

\medskip\noindent{\bf Sub-case $\times$.{B}.3: } $|H_R|>\phi|G|$.
Symmetrically construction $\times.B.(d)$ gives a type-$\underline{B}$
drawing of height $h$.

\medskip\noindent{\bf Case $\times$.C: } 
We want a type-C drawing of height $h+3$.

\medskip\noindent{\bf Sub-case $\times$.C.1: } $|H_M|\geq \tfrac{1}{2}|G|$.
We know that one of $H_L,H_R$ has size at most $\tfrac{1}{2}(|G|-|H_M|)\leq \tfrac{1}{4}|G|$.
Let us assume that this is $H_L$, and we then use construction
$\times.C.(a)$ from Figure~\ref{fig:timesC}
(the other case uses construction $\times.C.(b)$ and is similarly analyzed).

\begin{figure}[ht]
\subfigure[~]{\includegraphics[scale=0.38,page=9]{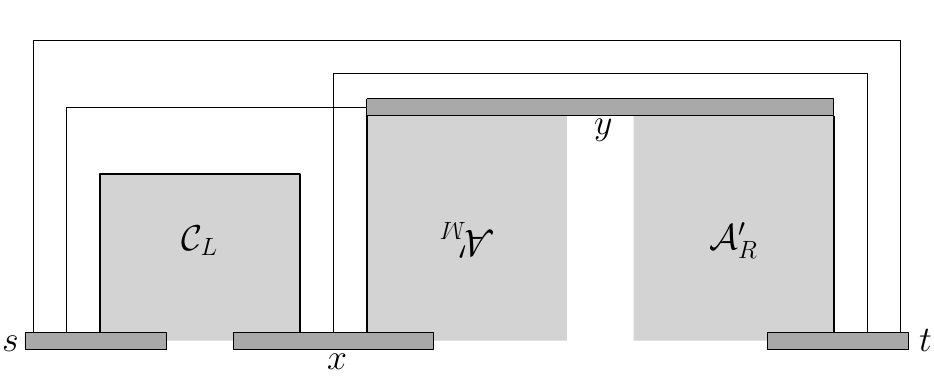}}
\hspace*{\fill}
\subfigure[~]{\includegraphics[scale=0.38,page=10]{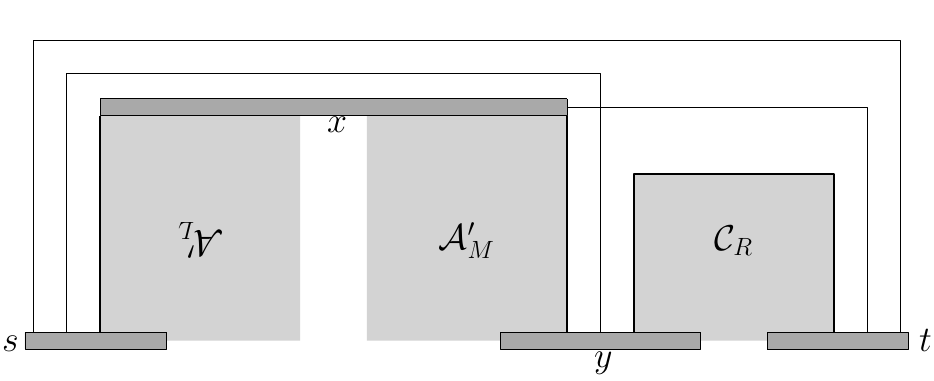}}
\newline
\newline
\subfigure[~]{\includegraphics[scale=0.38,page=12]{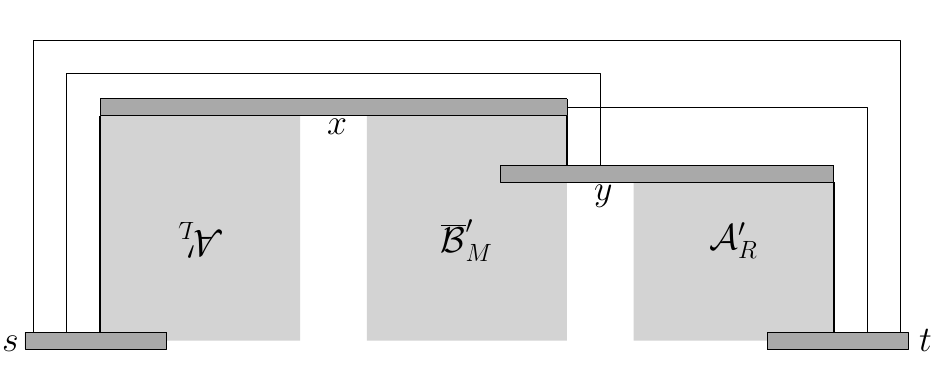}}
\hspace*{\fill}
\subfigure[~]{\includegraphics[scale=0.38,page=11]{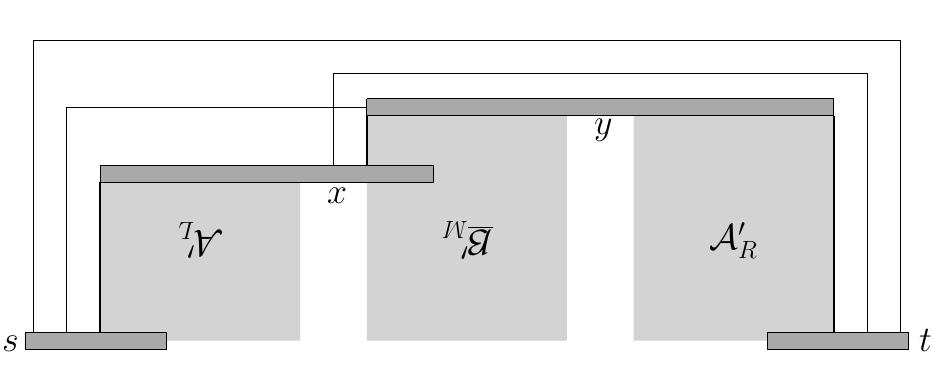}}
\caption{Constructions for $\times.C$.}
\label{fig:timesC}
\end{figure}

Drawings $\calA_M$ and $\calA_R$ both have height at most $h$,
and we need three more rows (one to release $y$, one for $(x,t)$
and one for $(s,t)$) so the height requirement here is $h+3$.
By $|H_L|\leq \tfrac{1}{4}|G|$, drawing $\calC_L$ has height at most
$$h(H_L)+3=\gamma \log |H_L| + \delta  + 3 \leq h(G) + \gamma \log (\tfrac{1}{4})+ 3 \leq  h-1$$
 by $\gamma\geq 2$.  We require three more rows above it: one for $(s,y)$,
one row that was used for $(x,t)$ elsewhere, and one row for $(s,t)$.  So the height requirement
here is actually strictly less than $h+3$.

\medskip\noindent{\bf Sub-case $\times$.C.2: } $|H_{M}|  \leq  \tfrac{1}{2} |G|$.
We know that one of $H_L,H_R$ has size at most $\tfrac{1}{2}|G|$.
Let us assume that this is $H_R$, and we then use construction
$\times.C.(c)$ from Figure~\ref{fig:timesC}
(the other case uses construction $\times.C.(d)$ and is similarly analyzed).

Drawing $\calA_L$ has height at most $h$, and we need three more rows (for releasing
$x$, edge $(s,y)$ and edge $(s,t)$), so the height requirement here is at most $h+3$.
Drawing $\overline{\calB}_M$ has height at most 
$$h(H_M) +2 
\leq h(G) + \gamma \log (\tfrac{1}{2})  +2 \leq h$$
by $\gamma\geq 2$.
Again we need three more rows, so the height requirement is at most $h+3$.
Similarly by $|H_R|\leq \tfrac{1}{2}|G|$ drawing $\calA_R$ has height at most $h-2$.
We need five more rows at $\calA_R$:  one row for releasing $y$,
one row because $\overline{\calB}_M'$ had one row between $y$ and the
(released) $x$, one row for $(x,t)$, one row that was used for $(s,y)$
elsewhere, and one row for $(s,t)$.   
So the height requirement everywhere is at most $h+3$.  

\subsection{Induction step---Case $\Delta$}

Now we turn our attention to the (much simpler) case $\Delta$ where the region incident
to edge $(s,t)$ has no crossing. We again distinguish cases by the drawing-type 
that we want to achieve.

\medskip\noindent{\bf Case $\Delta$.A: } 
We want a type-A drawing of height $h$.

\medskip\noindent{\bf Sub-case $\Delta$.A.1: } $|H_L|,|H_R|\leq \phi|G|$.
Then use construction $\Delta.A.(a)$ from Figure~\ref{fig:DeltaA}.
Drawing $\overline{\calB_L}$ has height at most 
$$ h(H_L)+2 = \gamma \log|H_L| + \delta + 2 \leq h(G) + \gamma \log\phi + 2 \leq h$$
since $\gamma \log \phi \leq -2$.
Similarly drawing ${\calA_R}$ has height at most $h-2$
and we need two more
rows above it
(for releasing $x$ and the bar of $s$).   So the height requirement is at
most $h$ everywhere.

\begin{figure}[ht]
\hspace*{\fill}
\subfigure[~]{\includegraphics[scale=0.38,page=18]{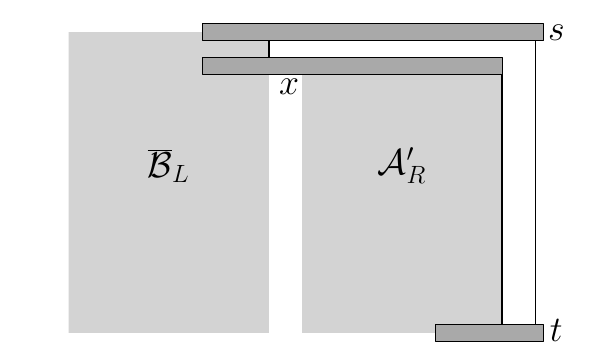}}
\hspace*{\fill}
\subfigure[~]{\includegraphics[scale=0.38,page=19]{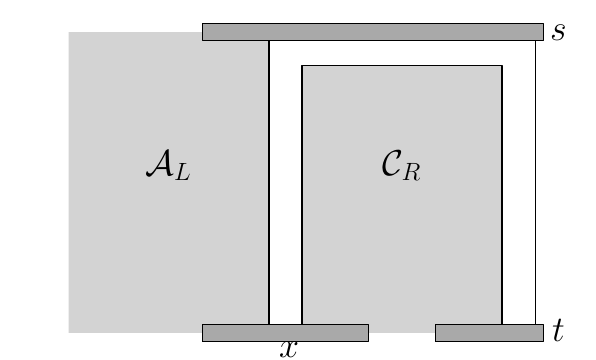}}
\hspace*{\fill}
\subfigure[~]{\includegraphics[scale=0.38,page=20]{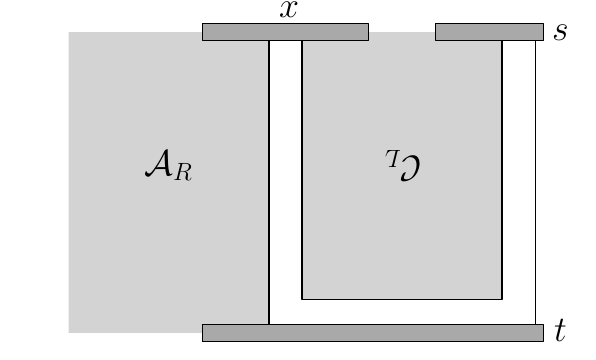}}
\hspace*{\fill}
\caption{Constructions for $\Delta.A$.}
\label{fig:DeltaA}
\end{figure}

\medskip\noindent{\bf Sub-case $\Delta$.A.2: } $|H_L|> \phi|G|$.
Then use construction $\Delta.A.(b)$.
Drawing $\calA_L$ has height at most $h$.
By $|H_R|\leq |G|-|H_L|\leq (1-\phi)|G|$,
drawing $\calC_R$ has height at most
$$h(H_R)+3 
\leq h(G) + \gamma \log (1-\phi)+ 3 \leq  h-1$$
since $\gamma \log (1-\phi)\leq -4$.
We require one more row above it (for the bar of $s$), so the height 
requirement is at most $h$ everywhere.

\medskip\noindent{\bf Sub-case $\Delta$.A.3: } $|H_R|>\phi|G|$.
Symmetrically one proves that construction $\Delta.A.(c)$ has height at most $h$.

\medskip\noindent{\bf Case $\Delta$.B: } 
We want two drawings, of type
$\overline{B}$, $\underline{B}$.  Both have height at most $h+2$,
and one has height at most $h$.

The construction for the type-$\overline{B}$ drawing is in Figure~\ref{fig:DeltaB}(a).
Since $\calA_L$ and $\calA_R$ have height at most $h$, and we need two
further rows above $\calA_R$, the height requirement is at most $h+2$ everywhere.  
If $|H_R|\leq \tfrac{1}{2} |G|$
then the height of $\calA_R$ is at most 
$$h(H_R) 
\leq h(G) +   \gamma \log \tfrac{1}{2} \leq h-2$$ 
by $\gamma\geq 2$ and so the height requirement is at most $h$ everywhere.

\begin{figure}[ht]
\hspace*{\fill}
\subfigure[~]{\includegraphics[scale=0.38,page=21]{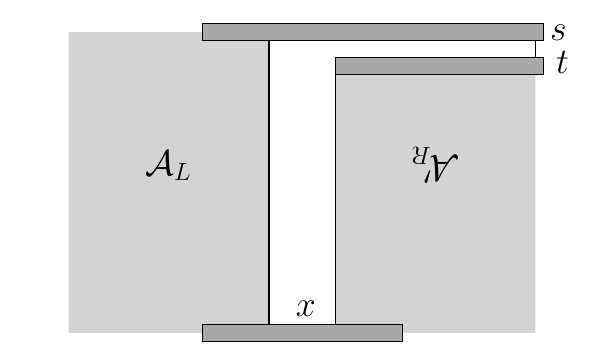}}
\hspace*{\fill}
\subfigure[~]{\includegraphics[scale=0.38,page=22]{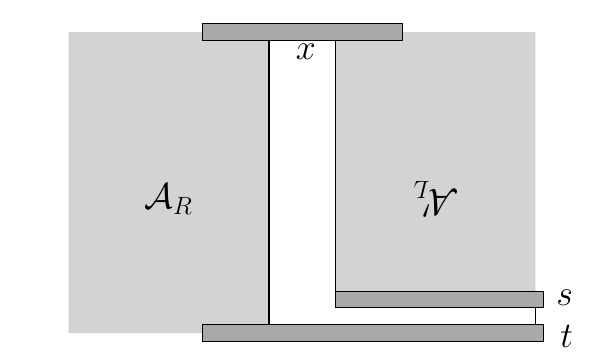}}
\hspace*{\fill}
\subfigure[~]{\includegraphics[scale=0.38,page=23]{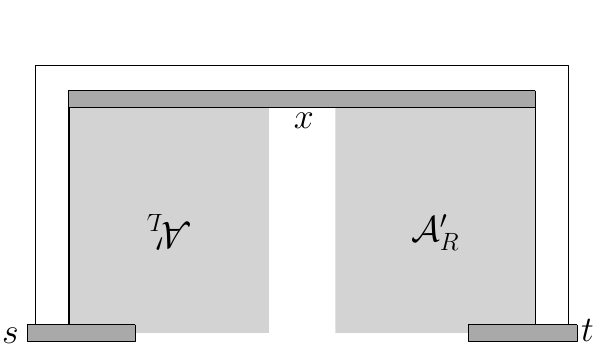}}
\hspace*{\fill}
\caption{Constructions for $\Delta.\overline{B}$ and $\Delta.\underline{B}$, as well as for $\Delta.C$.}
\label{fig:DeltaB}
\end{figure}

Likewise, the construction of a type-$\underline{B}$ drawing in Figure~\ref{fig:DeltaB}(b)
has height at most $h+2$, and if $|H_L|\leq \tfrac{1}{2} |G|$ then the height is at most $h$.
Since one of $H_L$ and $H_R$ has size at most $\tfrac{1}{2}|G|$,
one of the drawings has height at most $h$. 

\medskip\noindent{\bf Case $\Delta.C$: } 
We want a type-C drawing of height at most $h+3$.
The construction is shown in Figure~\ref{fig:DeltaB}(c).
Since $\calA_L$ and $\calA_R$ have height at most $h$, and we need two
more rows above them (to release $x$ and for edge $(s,t)$), the height is 
actually at most $h+2$. 

\subsection{Putting it all together}

We have given suitable constructions in all cases, so by induction
Lemma~\ref{lem:main} holds.  Using the type-A drawing, we get a
drawing of height $3.94\log(n-1)+2$ and width $O(n)$.  Therefore
Theorem~\ref{thm:main} holds.  Following the proof, one also sees
that the drawing can easily be found in linear time, since we can
construct the four drawings of each hanging subgraph in constant
time from the drawings of its subgraph.

\section{Conclusion}

In this paper, we pointed out an error in a result by Auer et al., and
show that for some outer-1-planar graphs, any poly-line  drawing without
crossings requires $\Omega(n^2)$ area.  We
then studied orthogonal box-drawings of outer-1-planar graphs that achieve
small area.  We create such drawings (using bars to represent vertices)
that have $O(n\log n)$ area and at most 2 bends per edge, and exactly
reflect the given outer-1-planar embedding.

We believe that reducing the number of bends per edge should be possible,
and in particular, conjecture that we can achieve $O(n\log n)$ area with
at most one bend per edge, perhaps at the expense of modifying the 1-planar
embedding.  On the other hand, finding drawings with 0 bends (i.e.,
visibility representations) appears difficult.  Perhaps 
{\em bar-1-visibility drawings} (where edges are allowed to go through
up to one bar of a vertex) may be possible while keeping the area
sub-quadratic.

Straight-line drawings of outer-1-planar graphs are also of interest.
It is known that there are order-preserving outer-1-planar straight-line
drawings of area $O(n^2)$ \cite{Auer16}.  Are there straight-line
drawings of sub-quadratic area (again perhaps at the expense of not
respecting the 1-planar embedding)?

\fi
\bibliographystyle{plain}
\bibliography{full,gd,papers}

\begin{thebibliography}{10}

\bibitem{Argyriou}
E.~Argyriou, S.~Cornelsen, H.~F{\"{o}}rster, M.~Kaufmann, M.~N{\"{o}}llenburg,
  Y.~Okamoto, C.~Raftopoulou, and A.~Wolff.
\newblock Orthogonal and smooth orthogonal layouts of 1-planar graphs with low
  edge complexity.
\newblock In T.~Biedl and A.~Kerren, editors, {\em Graph Drawing and Network
  Visualization (GD 2018)}, volume 11282 of {\em LNCS}, pages 509--523.
  Springer, 2018.

\bibitem{Auer16}
C.~Auer, C.~Bachmaier, F.~Brandenburg, A.~Glei{\ss}ner, K.~Hanauer,
  D.~Neuwirth, and J.~Reislhuber.
\newblock Outer 1-planar graphs.
\newblock {\em Algorithmica}, 74(4):1293--1320, 2016.

\bibitem{Bie-GD02}
T.~Biedl.
\newblock Drawing outer-planar graphs in {$O(n\log n)$} area.
\newblock In S.~Kobourov and M.~Goodrich, editors, {\em Graph Drawing (GD'01)},
  volume 2528 of {\em LNCS}, pages 54--65. Springer, 2002.

\bibitem{Bie-DCG11}
T.~Biedl.
\newblock Small drawings of outerplanar graphs, series-parallel graphs, and
  other planar graphs.
\newblock {\em Discrete and Computational Geometry}, 45(1):141--160, 2011.

\bibitem{Bie-GD14}
T.~Biedl.
\newblock Height-preserving transformations of planar graph drawings.
\newblock In C.~Duncan and A.~Symvonis, editors, {\em Graph Drawing (GD'14)},
  volume 8871 of {\em LNCS}, pages 380--391. Springer, 2014.

\bibitem{Bie-GD20}
T.~Biedl.
\newblock Drawing outer-1-planar graphs revisied.
\newblock In {\em Graph Drawing and Network Visualization (GD'20)}, LNCS.
  Springer, 2020.
\newblock {Poster with a short abstract. To appear. See also ArXiV 2009.07106.}

\bibitem{DE12}
H.~Dehkordi and P.~Eades.
\newblock Every outer-1-plane graph has a right angle crossing drawing.
\newblock {\em Int. J. Comput. Geom. Appl.}, 22(6):543--558, 2012.

\bibitem{Die12}
R.~Diestel.
\newblock {\em Graph Theory, 4th Edition}, volume 173 of {\em Graduate texts in
  mathematics}.
\newblock Springer, 2012.

\bibitem{Eggleton}
R.~Eggleton.
\newblock Rectilinear drawings of graphs.
\newblock {\em Utilitas Mathematica}, 29:149--172, 1986.

\bibitem{FPP88}
H.~de Fraysseix, J.~Pach, and R.~Pollack.
\newblock Small sets supporting {Fary} embeddings of planar graphs.
\newblock In {\em ACM Symposium on Theory of Computing (STOC '88)}, pages
  426--433, 1988.

\bibitem{HEL+12}
S.~Hong, P.~Eades, G.~Liotta, and S.~Poon.
\newblock F{\'{a}}ry's theorem for 1-planar graphs.
\newblock In J.~Gudmundsson, J.~Mestre, and T.~Viglas, editors, {\em Computing
  and Combinatorics (COCOON 2012)}, volume 7434 of {\em LNCS}, pages 335--346.
  Springer, 2012.

\bibitem{KLM17}
S.~Kobourov, G.~Liotta, and F.~Montecchiani.
\newblock An annotated bibliography on 1-planarity.
\newblock {\em Computer Science Review}, 25:49--67, 2017.

\end{thebibliography}

\end{document}